\documentclass{article}

\usepackage{arxiv}
\usepackage{amsmath,amsfonts}
\usepackage{algorithmic}
\usepackage{array}
\usepackage[caption=false,font=normalsize,labelfont=sf,textfont=sf]{subfig}
\usepackage{textcomp}
\usepackage{stfloats}
\usepackage{url}
\usepackage{verbatim}
\usepackage{multirow}

\usepackage{graphicx}
\usepackage{cite}
\usepackage[ruled,vlined,noend,linesnumbered]{algorithm2e}

\newtheorem{theorem}{Theorem}

\newtheorem{proposition}{Proposition}

\newtheorem{proof}{Proof}

\begin{document}

\title{Acoustic Model Optimization over Multiple Data Sources: Merging and Valuation}


\author{Victor Junqiu Wei, Weicheng Wang \\
Department of Computer Science and Engineering\\
The Hong Kong University of Science and Technology\\
\texttt{wjqjsnj@gmail.com, wwangby@connect.ust.hk} \\
\And
Di Jiang, Conghui Tan, Rongzhong Lian\\
AI Group, WeBank Co., Ltd, Shenzhen, China\\
\texttt{\{dijiang,martintan,ronlian\}@webank.com}
}



\maketitle

\begin{abstract}
Due to the rising awareness of privacy protection and the voluminous scale of speech data, it is becoming infeasible for Automatic Speech Recognition (ASR) system developers to train the acoustic model with complete data as before. For example, the data may be owned by different curators, and it is not allowed to share with others. In this paper, we propose a novel paradigm to solve salient problems plaguing the ASR field. In the first stage, multiple acoustic models are trained based upon different subsets of the complete speech data, while in the second phase, two novel algorithms are utilized to generate a high-quality acoustic model based upon those trained on data subsets. We first propose the Genetic Merge Algorithm (GMA), which is a highly specialized algorithm for optimizing acoustic models but suffers from low efficiency. We further propose the SGD-Based Optimizational Merge Algorithm (SOMA), which effectively alleviates the efficiency bottleneck of GMA and maintains superior model accuracy. Extensive experiments on public data show that the proposed methods can significantly outperform the state-of-the-art. Furthermore, we introduce Shapley Value to estimate the contribution score of the trained models, which is useful for evaluating the effectiveness of the data and providing fair incentives to their curators.
\end{abstract}


\section{Introduction}

\let\thefootnote\relax\footnotetext{This paper is an extended version of \cite{tan2020novo}.}

{A}{utomatic} Speech Recognition (ASR) has already become an indispensable part of modern intelligence systems such as voice assistants and client service robots. 
An effective ASR system relies on a robust acoustic model that is trained over a huge amount of speech data collected from a wide range of domains. However, in real-life scenarios, training an acoustic model with complete data is increasingly infeasible due to the following three reasons:
\begin{itemize}
    \item \textbf{R1:} Speech data from different domains are owned by distinct curators, who are usually unwilling to share these data due to privacy concerns.
    Meanwhile, to encourage curators to contribute to the training process, providing fair incentives for their participation is also challenging.
    
    \item \textbf{R2:} Speech data of multiple curators may be distributed across different computing centers. Traditional distributed computing paradigms, e.g., ParameterServer, require frequent information exchange between different computing nodes. 
    Thus, if the data are remotely distributed, the information exchange is problematic.  
    \item \textbf{R3:} The speech data from a single curator can be voluminous. Traditional optimizing methods like Stochastic Gradient Descent (SGD) are slow, even for their parallel or asynchronous versions. Instead, one would like to process data in a fully parallel way.
\end{itemize}

To tackle the issues mentioned above, we propose to optimize the acoustic model in a two-stage fashion: (1) we first train models on different parts of the data independently until they converge on local data; (2) after local training, we gather all separately trained models and merge them into a single one. By this means, only one round of communication is required to train the model.  
This technique has already been applied in some existing applications \cite{mcmahan2016communication, povey2011kaldi, povey2014parallel}. 
For example, \cite{mcmahan2016communication} proposed to use this technique to solve the privacy and communication problems (i.e., \textbf{R1} and \textbf{R2}) in federated learning. 
Kaldi \cite{povey2011kaldi}, one of the most widely-used speech recognition toolkits, has adopted it as the default training scheme to deal with the efficiency challenge (i.e., \textbf{R3}) \cite{povey2014parallel}. 

However, in terms of merging models, most of the existing work still relies on the simplistic technique of averaging over all the models. 
Although some progress has been made in improving performance, a better model merging strategy is still an open problem. 
Besides, since the training data are collected from multiple curators, it is important to fairly measure the contributions of different data, e.g., to offer the curators corresponding rewards. 
Nevertheless, evaluating participants' contributions is currently under-exploited. \looseness=-1

In this paper, we propose a new optimization paradigm. Unlike existing studies, we merge the models in a more data-efficient way, which improves the model quality with a limited amount of data.
Moreover, the Shapley Value is introduced to evaluate the contribution of the data provided by each curator to the target model so that curators can achieve fair incentive rewards. 
In detail, we first propose a merge method based on the genetic algorithm, named GMA, which is capable of yielding models of great recognition accuracy. 
However, its practicality is heavily limited by its poor efficiency. 
To further tackle this issue, we convert the model merging problem into a mathematical optimization problem via a novel formulation and develop a new optimization method based on SGD to solve it, which leads to a new method called SOMA. 

This paper differs from the preliminary conference version~\cite{tan2020novo} in the following ways. 
First of all, in order to evaluate the contribution of different data curators and offer fair rewards
to them, we have proposed a novel method concerning the valuation of the source models which is present in a newly added section (Section~\ref{sec:valuation}). 
We also conducted experiments on the valuation of the source models in Section~\ref{sec:exp-value}. 
Secondly, an extra group of experiments has been conducted, where we empirically evaluated the co-contribution scores and studied their relationships to other statistical quantities. 
Thirdly, the related work section (Section~\ref{sec:related}) has been expanded. 
We reviewed the related literature more comprehensively, and the works related to incentive mechanisms are also included. 
Fourthly, more details have been added to some existing technical parts to make them easier to read.


Experiments suggest SOMA can produce models comparable to GMA, but with much less computational cost. 
Without loss of generality, we focus on the scenario of merging several DNN acoustic models with homogeneous structures into a single one. 
This practice can be straightforwardly applied in the state-of-the-art ASR systems based upon DNN-HMM or End-to-End architectures.

\section{Related Work}
\label{sec:related}
\subsection{Automatic Speech Recognition}
Automatic speech recognition (ASR) is a technology for smoother interaction between people and machines. Its applications include mobile speech recognition technology, smart speakers, and other smart voice products. The ASR converts speech fragments into text, which typically includes an acoustic model (AM) and a language model (LM). Recently, there has been some research on end-to-end ASR that does not build AM and LM explicitly. However, traditional ASR systems are still dominating in practice due to their high performance and reliability.

The acoustic model aims to find the probability distribution on a phoneme sequence given a piece of speech. Before the rise of deep learning, practical acoustic models were implemented with statistical models such as the classical Gaussian mixture model (GMM) and the Hidden Markov model (HMM). GMM and HMM have a critical position due to their mathematical elegance and capability to model time-varying sequences \cite{gales2008application}. However, they fail to capture complex sequential information and non-linear speech features.
With the rise of deep learning, deep learning-based acoustic models, such as CNN, LSTM, encoder-decoder frameworks, and attention mechanisms, have further improved performance, becoming dominating in the field of speech recognition.

Subsequent researchers found discriminative training can get relatively better results \cite{waibel1989phoneme}\cite{jiang2023probabilistic}\cite{song2021l2rs}. Some research proposed better feature extraction models to replace GMM, including neural network (NN), restricted Boltzmann machine (RBM), deep belief network (DBN), and deep neural network (DNN) \cite{purwins2019deep}. The outstanding performance of the hybrid model has also attracted much attention.

Deep learning methods enhance the speech signal representation compared to traditional acoustic models such as GMM. Inspired by the GMM-HMM structure, researchers propose to replace GMM with DNN, resulting in a popular DNN-HMM hybrid system, which utilizes the advantages of modeling sequential information of HMM and the superior ability to capture speech representation. In 2012, DNN trained on very large-scale data successfully reduced the word error rate (WER)\cite{hinton2012deep}, stressing the paramount potential of DNN's ability to learn the hierarchical structure of the representation from the input data. Further advanced deep models, such as the recurrent neural network (RNN) (including long short-term memory (LSTM) and gated recurrent unit (GRU)) and convolutional neural network (CNN), quickly surpassed DNN. They enhance the ability to capture the rich structural information from speech \cite{latif2018phonocardiographic,chen2023neural,hong2024infantcrynet,song2022platform,wu2023enhance,chen2021scalable,jiang2021gdpr}. The lack of annotated data promotes unsupervised representation learning research. For the unsupervised representation learning of speech, the autoencoder (AE), the restricted Boltzmann machine (RBM), and the deep belief network (DBN) are widely used \cite{langkvist2014review}.

With the recent interest in generative models, variational auto-encoder (VAE), generative adversarial network (GAN), and deep autoregressive models are also used in AM \cite{van2016pixel,deepmind2016generative}. Among them, VAE and GAN can achieve unsupervised acquisition of speech features \cite{bollepalli2019generative,hsu2017unsupervised}. Combined with the characteristics of the data in practical applications, the feature capture capabilities of deep learning under different settings have helped speech recognition. 

The adaptation of the acoustic model is also of our interest in this paper. The existing method of AM adaptation focus on addressing the feature shift (covariate shift) of speech, such as speaker, environmental noise, pitch, loudness, etc. \cite{ghorbani2019domain}. While in federated learning scenarios, we aim to merge the AM from curators with different distribution to form a stronger AM that fits all participants. 

\subsubsection{Language model}

In the ASR system, the language model (LM) is another critical component that facilitates converting the output of AM to a logical and natural sentence. Language models have achieved significant performance improvements with natural language processing development,  leading to better ASR performance. 

The primary language model is the bag-of-words model (BoW), a one-hot form of text representation. This method is suitable for processing discrete data and extending features but does not consider the order between words\cite{baeza2015predicting}. Although this type of method is simple, the semantic representation is not exact. In ASR, the most commonly used language model is $N$-gram, which models the next-word probability by counting the co-occurrence. Simple as $N$-gram is, its robustness and interpretability make it a worthy choice as the LM in ASR systems.

In contrast to statistical language models, Begio et al. proposed the concept of neural network language models in 2003 \cite{bengio2003neural}. The performance of the language model has been effectively improved with the emergence of the word vector model. The most representative ones are Word2Vec \cite{mikolov2013efficient}, and Glove \cite{pennington2014glove}, whose expression is to convert each word into a vector representation with richer semantic information. As contextual words are very informative, researchers \cite{wang2018densely} propose to apply CNN instead of $N$-gram to capture context from a larger receptive field. Meanwhile, the success of RNNLM shows the significant importance of modeling long-range sequential context.

On the other side, pre-training models have now become the main idea of language models~\cite{jiang2023probabilistic, lu2024pretraining,wu2022phonetic,zhou2018multi,devlin2018bert,radford2019language,lan2019albert,liu2019roberta,wei2021training,wei2019nezha,li2021heterogeneous,jiang2021industrial,zhou2021memetic,jiang2019federated, song2020topicocean,hong2024expanding,jiang2016latent}. ELMo \cite{peters2018deep} uses LSTM structure and two-way settings to fully consider contextual information and reflect the characteristics of different dimensions. The pre-trained model can cope with different downstream tasks and has a better overall performance. Transformer \cite{vaswani2017attention} (and its variants~\cite{zhang2021continuous,zhang2020tensorcoder,wang2022clusterformer,li2022hypoformer}), a powerful attention-based architecture, has been proved of its overwhelming ability in modeling deep structural information from data \cite{radford2018improving,zhou2018multi}. Transformer is soon applied to large-scale pre-trained language models, including GPT \cite{radford2018improving}, BERT \cite{devlin2018bert}, GPT2 \cite{radford2019language} and subsequent ALBERT \cite{lan2019albert}, RoBERTa \cite{liu2019roberta}, etc. However, it may be inefficient to apply pre-trained language models such as BERT and GPT to LM in ASR \cite{DBLP:journals/corr/abs-1904-09408}. 
Therefore, it is an open problem on how to utilize and merge LM in federated ASR.

\subsubsection{End-to-end ASR model}

Some researchers consider ASR as a sequence recognition task. Inspired by other sequence modeling tasks such as dialogue generation and machine translation, it is reasonable that ASR performance can be improved by replacing the module-based approach with an end-to-end manner, which avoids error accumulation between modules. The end-to-end ASR model is elegant and straightforward without introducing noises from AM and LM, which consist of an encoder-decoder structure based on attention mechanism \cite{sutskever2014sequence,bahdanau2014neural}.

However, end-to-end ASR is still challenging in real applications. For example, errors in previous time-steps may propagate to subsequent decisions, especially when the input speech is long. The mismatch between training and evaluation objects usually brings suboptimal performance \cite{wu2016google}. Tjandra et al. \cite{tjandra2018sequence} propose to integrate sequence-to-sequence approaches with reinforcement learning to tackle the gap between training object and evaluation utility. Li et al. \cite{li2019jasper} introduce a new layer-wise optimizer for the end-to-end speech recognition model, which enhances training convergence. Pham et al. \cite{pham2019very} propose use self-attention via the Transformer architecture to replace the traditional end-to-end model and shows better performance than previous approaches. However, as end-to-end ASR systems consume many computational resources, it is not suitable to deploy on edge devices and in a federated learning scenario.

\subsection{Acoustic Model Optimization}

Existing optimization methods of training deep neural networks can be divided into three categories: first-order optimization methods represented by the widely used stochastic gradient methods (SGD), high-order optimization methods such as Newton's method, and heuristic optimization methods such as the coordinate descent method \cite{sun2019survey}.

Among them, SGD \cite{robbins1951stochastic,jain2018parallelizing} are widely used for training acoustic model \cite{yu2016automatic}. With the characteristic that each iteration is independent of the total amount of data, SGD has achieved excellent performance in the training of sparsely distributed data \cite{konevcny2016federated}. Furthermore, the variants of SGD such as Adam \cite{kingma2014adam} get a lot of attention and appreciation. Stochastic averaging gradient (SAG) and stochastic variance-reduced gradient (SVRG) based on gradient descent are proposed to solve the problem of the slow convergence speed. Experiments in \cite{reddi2016stochastic,babanezhad2015stop,johnson2013accelerating} suggest that SVRG matches basic SGD, but it needs to process the whole dataset in a number of iterations. Researchers propose a scheme for applying SVRG in a distributed setting and observe its better performance. Nevertheless, it is difficult to implement it due to the operations on a large amount of data \cite{lee2015distributed,lee2017distributed}. 

Researches in the heuristic optimization method also have made progress. Talbi et al. propose to use variation operators (e.g., mutation, crossover) to select and reproduce DNNs on their representations \cite{talbi2020optimization}. Cui et al. proposed to use a combinational optimization strategy with SGD algorithm and evolutionary learning \cite{cui2019acoustic}. As for the distributed computation setting, asynchronous SGD under the ParameterServer framework is widely adopted \cite{zhang2013asynchronous}.

Moreover, Taming et al. \cite{de2015taming} propose an asynchronous SGD (A-SGD) algorithm using lower-precision arithmetic, avoiding a variety of problems on modern hardware while maintaining good performance SGD. A-SGD is further applied to industrial applications composed of large deep neural networks, improving those large distributed machine learning clusters' performance \cite{dean2012large,chilimbi2014project}. However, the limitation of the communication conditions between the participants and the central server prevents the federated system from achieving the best performance with conventional optimization methods. Some previous works help relieve the aforementioned limitations \cite{kumar2019static,meinedo2000combination,xiong2018microsoft}. Povey et al., \cite{povey2014parallel} proposed a method of training acoustic models on subsets of data independently and then merging them periodically. 

Model combination in previous works \cite{kumar2019static,meinedo2000combination,xiong2018microsoft} does not merge the models into a single one in the training phase, where all the models are respectively evaluated, and their outputs are combined to produce the final results during inference. Hence, they are more related to ensemble learning \cite{zhou2012ensemble}. Model ensembling increases the burden of prediction, such as computational cost and communication rounds, in the federated learning scenario. 

\subsection{Incentive Mechanism in Federated Scenario}

Researchers have proposed many incentive mechanisms deliberately designed for FL, which reward and react to these participants with different quality based on contribution evaluation to encourage more participants to contribute their high-quality data to the FL process. Researchers leverage the idea from game theory to address this problem \cite{huang2020exploratory, weng2020fedserving, nishio2020estimation, hasan2021incentive}.

Some researchers apply Stackelberg games-based incentive mechanism to FL since the parameter server and participants in FL can be regarded as leaders and followers in Stackelberg games. Stackelberg games are used to analyze the federated scenario and to find the benefit equilibria between the organizer and the participants \cite{pandey2019incentivize, hu2020trading}. Sarikaya et al. \cite{sarikaya2019motivating} proposed to mitigate the delays in completion of each training batch by analytically obtaining an equilibrium solution of a Stackelberg game. Khan et al. \cite{khan2020federated} model the incentive-based interaction between a global server and participants for FL via a Stackelberg game to motivate the participants to join in the FL process. However, since the participants in FL do not upload data but upload parameters instead,  the utility functions are invisible to the server. Thus, Stackelberg games only work to incentivize network resources' contribution but can not calculate each participant's contribution to the model effect.

Some researchers use incentive mechanisms inspired by Contract Theory to motivate participants to join in FL \cite{kang2019incentive, ding2020incentive}. Kang et al. \cite{kang2019incentive2} accomplish the pairwise contribution qualification by introducing reputation and let participants evaluate each other, and combine contract theory to motivate participants with high-quality data. A hierarchical incentive mechanism design for FL is proposed in \cite{lim2020hierarchical}, which considers multiple model owners and the formation of multiple federations. It analyzes the equilibrium that the model reaches after iterations of merges and splits, applying solutions from Stackelberg games. For a fair distribution of payoffs, they adopt a coalitional game approach based on each model owner's marginal contribution to the federation. Zhan et al. \cite{zhan2020learning} design a deep reinforcement learning-based (DRL) incentive mechanism for FL to motivate edge nodes to contribute model training. VCG-based FL incentive mechanisms are designed for incentivizing data owners to contribute all their data and truthfully report their costs in FL settings \cite{cong2020vcg, cong2020game}. Zeng et al. \cite{zeng2020fmore} and Le et al. \cite{le2020incentive} formulate the incentive mechanism between the participants of federated learning as an auction. Yu et al. propose the FL incentivizer (FLI), which dynamically divides a given budget in a context-aware manner among data owners in a federation by jointly maximizing the collective utility while minimizing the inequality among the data owners, in terms of the payoff received and the waiting time for receiving payoffs \cite{yu2020sustainable, yu2020fairness}. Recently, researchers introduce Shapley Value or its approximation to evaluate the contributions of participants in FL \cite{wang2019measure, song2019profit}.

\section{Problem Setup}

Assume that we have $n$ acoustic models $\{M_{S,1}, M_{S,2}, \dots, \\ M_{S,n}\}$ with the same structures but different parameters since they are trained on different data. We call these $n$ acoustic models as the \emph{source models}. Our goal is to merge them into one \emph{target model} $M_T$, which possesses the same structure as source models but has better performance.

For each acoustic model $M_i$, we assume it has $L$ DNN layers, and its parameter of $l$-th layer is denoted as $W_i^l$ ($1\leq l\leq L$). $W_i^l$ includes all types of trainable parameters on that layer, such as weight and bias. For notational simplicity, we use the operation on model $M_i$ to denote the operation on all its parameters $W_i^l$ with $l=1,\dots,L$. For example, $(M_i+M_j)/2$ is the model generated by averaging all the corresponding parameters of $M_i$ and $M_j$.

In order to select the best $M_T$, we need some data to evaluate the quality of $M_T$, and we refer those data as the validation data. In contrast, the data used for training $\{M_{S,1},M_{S,2},\dots,M_{S,n}\}$ are called as the training data. Though extra validation data is required for our scheme, it will be shown later that the model quality can be greatly improved with very little validation data, which implies it is possible to obtain acoustic models of similar performance with less training data. To measure the performance of acoustic models, we utilize the widely used metric Word Error Rate (WER), which is defined as the minimum edit distance between the ASR hypothesis and the ground truth over the number of words in the ground truth.

\section{Genetic Merge Algorithm}

Since we aim at discovering a better target acoustic model from a set of source models, a relatively straightforward approach is to consider the source models as the initial population and apply the genetic algorithm \cite{holland1992adaptation}. Genetic algorithms are a class of heuristic search algorithms inspired by biological evolution. It optimizes a group of candidates by repeatedly generating new individuals via operations like mutation and crossover and then offers the ones with large fitness the right to produce offsprings, just like how biological evolution works. A genetic algorithm has two important factors that determine its performance: the scheme of generating offsprings and the strategy of selecting the fittest individuals. In the following, we propose the Genetic Merge Algorithm (GMA), which is calibrated for the scenario of acoustic model optimization in ASR systems.

In GMA, the scheme for generating offspring includes four different operators. The first three are classical in GA: reproduction, mutation and crossover. Specifically, we choose to use single point mutation operator and the one-point crossover operator respectively. Moreover, inspired by the phenomenon discovered in \cite{mcmahan2016communication} that directly averaging two neural network models with the same initialization but trained on different data can lead to a better one, we propose a new operator called linear interpolation operator. In detail, these four operators work as follows:
\begin{itemize}
\item \emph{Reproduction} directly copies the existing models into next generation.
\item \emph{Mutation} randomly changes one bit in the binary expression of the parameters for the selected model.
\item \emph{Crossover} takes two parent models as the input. It randomly draws an integer $l$ ($1\leq l<L$), and the first $l$ layers of these two models are swapped. For example, if $M_1$ and $M_2$ are two input parent models, then two generated offsprings are:
    \begin{align*}
        M_{\text{new},1} =& \{W_1^1,\dots,W_1^{l}, W_2^{l+1}, \dots, W_2^L\}, \\
        M_{\text{new},2} =& \{W_2^1,\dots,W_2^{l}, W_1^{l+1}, \dots, W_1^L\}.
    \end{align*}
\item \emph{Linear interpolation} linearly combines all the parameters of two parent models in a weighted way to generate one new model, i.e.,
    \[
        M_{\text{new}} = \lambda M_1 + (1 - \lambda) M_2,
    \]
    where $\lambda$ is an interpolation coefficient randomly sampled from $(0,1)$. Obviously, this operator is an extension of simple average.
\end{itemize}
For the normalization layers like batch normalization, we view the statistical variables on these layers (such as the batch mean and variance in batch normalization) as parameters, and apply genetic operators to these statistical variables.

We utilize WER as the measurement of the fitness and lower WER implies better fitness. For each generation, we evaluate the WER of each acoustic model on the validation set and choose the top-$K$ with the lowest WERs as the parents of the next generation, where $K$ is a hyper-parameter. Furthermore, inspired by the fact that the simple average of all the source models is already a good choice for $M_T$ again, we also include the averaged model $\sum_{i=1}^n M_{S,i}/n$ into the initial population. Such initialization provides a better start for GMA and reduces the time needed for converging. Besides, it ensures that the final acoustic model generated by GMA is always better than the simple average. 

The complete workflow of GMA is presented in Algorithm~\ref{alg:ga}. Besides $K$, GMA requires three additional hyperparameters: $p_1$, $p_2$ and $p_3$, which are the probabilities that mutation, crossover and linear interpolation operators are applied to generate offsprings respectively. It is worth noting that larger $p_1,p_2,p_3$ and $K$ bring more diversity into the population and usually lead to better searching results, but they also increase computation costs.

\begin{algorithm}[h]
    \SetKwInOut{Input}{input}
    \SetKwInOut{Output}{output}
    \SetAlgoLined
    \Input{source models $M_{S,1},M_{S,2},\dots,M_{S,n}$}
        Initialize $P=\{M_{S,1},\dots,M_{S,n}, \sum_{i=1}^n M_{S,i}/n\}$ \\
        \While{not converged}{
        $P'=P$ \tcp*[f]{Reproduction} \\
            \ForEach{$M_i$ in $P$}{
            With probability $p_1$ let $P'=P'\cup \text{Mutation}(M_i)$
            }
            Randomly shuffle $P$ \\
            \ForEach{adjacent models $M_i, M_{i+1}$ in $P$}{
            With probability $p_2$ let $P'=P'\cup \text{Crossover}(M_i,M_{i+1})$ \\
            With probability $p_3$ let $P'=P'\cup \text{LinearInterpolation}(M_i,M_{i+1})$
            }
        Compute the WERs of the models in $P'$ on validation set\\
        Let $P$ be the set of top-$K$ models in $P'$
        }
        \Output{$M_T=$model in $P$ with lowest WER}
        \caption{Genetic Merge Algorithm (GMA)}
        \label{alg:ga}
\end{algorithm}

\section{SGD-Based Optimizational Merge Algorithm}

With its superior performance in generating high-quality acoustic models, GMA suffers from extremely low efficiency. Processing a few acoustic models with GMA on a small validation data (e.g., 5 source models on the validation set containing 10 hours of speech data) already requires several days, making it hardly applicable real-life applications.

To tackle this issue, we propose the SGD-Based Optimizational Merge Algorithm (SOMA), which enjoys similar performance as GMA but much more efficient. The major challenge of developing this method is how to convert our acoustic model merging problem into a mathematical optimization problem that SGD can be applied.

The key observation is that any model $M$ generated by GMA can be layer-wisely presented by the following formula:
\begin{equation}
\label{eq:def_mt}
\begin{aligned}
W^l=&\sum_{i=1}^n \theta_i^l W_{S,i}^l + \Delta W^l \\
\text{s.t.\ \ }& \theta_i^l\geq 0,\ \sum_{i=1}^n\theta_i^l=1
\end{aligned}
\end{equation}
for all its layers $W^l$ ($1\leq l\leq L$). Here the summation term $\sum_{i=1}^n \theta_i^l W_{S,i}^l$ corresponds to the linear interpolation and crossover operators, while the extra variable $\Delta W^l$ catches the change introduced by mutation. This fact is rigorously justified by the following proposition:

\begin{proposition}
    $M$ in the form of \eqref{eq:def_mt} covers any model generated by GMA. Besides, term $\Delta W^l$ is brought by the mutation operation. In other words, it always holds that $\Delta W^l=0$ if the mutation operator is not applied.
\end{proposition}
\begin{proof}
    We prove this proposition by induction.

    For any source model $M_{S,i}$, it is obvious that \eqref{eq:def_mt} can recover it by setting $\Delta W^l=0$, $\theta_i^l=1$ and all other $\theta_j^l=0$ with $j\neq i$. For the simple average $\sum_{i=1}^n M_{S,i}/n$, we can choose $\theta_i^l=1/n$ for all $i$ and also $\Delta W^l=0$ for all layers. Hence, all the models in the initial generation of GMA can be presented by \eqref{eq:def_mt} with $\Delta W^l=0$.

    Assume all the models in one generation satisfies equation \eqref{eq:def_mt}, then we proceed to show the proposition also holds for the next generation:
    
    It is trivial for the model generated by reproduction, since it is exactly the same as its parent.
    
    For the model generated by the mutation operator, it can be formulated into \eqref{eq:def_mt} by changing the variable $\Delta W^l$ of its parent while inheriting the other terms. 
    
    Since the crossover operator just swaps some layers of two models, and the internal structures of each layer is still preserved, while the constraints in \eqref{eq:def_mt} are imposed layer-wisely, the resulting models should still also stick to the pattern of \eqref{eq:def_mt} once their parents do, but with the corresponding $\theta_i^l$ and $\Delta W^l$ swapped. 
    
    As for the linear interpolation operator, assume $M_1$ and $M_2$ are two input parents, then their layers can be written as:
    \begin{align*}
        W_1^l=\sum_{i=1}^n\alpha_i^l W_{S,i}^l + \Delta W^l_1, \\
        W_2^l=\sum_{i=1}^n\beta_i^l W_{S,i}^l + \Delta W^l_2,
    \end{align*}
    where $\alpha_i^l$ and $\beta_i^l$ are two groups of realization for $\theta_i^l$ and satisfy the constraints for $\theta_i^l$ in \eqref{eq:def_mt}. Thus, the parameter of the linearly interpolated model should be:
    \begin{align*}
        W_{\text{new}}^l
        =& \lambda W_1^l + (1-\lambda) W_2^l \\
        =& \lambda\big(\sum_{i=1}^n\alpha_i^l W_{S,i}^l + \Delta W^l_1\big)
            + (1-\lambda)\big(\sum_{i=1}^n\beta_i^l W_{S,i}^l + \Delta W^l_2\big) \\
        =& \sum_{i=1}^n \underbrace{\big( \lambda \alpha_i^l + (1-\lambda)\beta_i^l \big)}_{\theta^l_{\text{new},i}} W^l_{S,i}
            + \underbrace{\lambda\Delta W^l_1 + (1-\lambda)\Delta W^l_2}_{\Delta W^l_{\text{new}}}.
    \end{align*}
    In the next, we need to show $\theta_{\text{new},i}^l$ satisfies the constraints in \eqref{eq:def_mt}. The non-negativity of $\theta_i^l$ is obvious, considering that $\alpha_i^l\geq 0$, $\beta_i^l\geq 0$ and $\theta\in(0,1)$. Besides,
    \begin{align*}
         \sum_{i=1}^n\theta_{\text{new},i}^l
        =& \sum_{i=1}^n \big( \lambda \alpha_i^l + (1-\lambda)\beta_i^l \big) \\
            =& \lambda \sum_{i=1}^n \alpha_i^l + (1-\lambda)\sum_{i=1}^n \beta_i^l \\
            =& \lambda + (1-\lambda) \\
            =& 1,
    \end{align*}
    where the facts $\sum_{i=1}^n \alpha_i^l=1$ and $\sum_{i=1}^n \beta_i^l=1$ are applied to the second equality. Hence, the model generated by linear interpolation also satisfies the desired pattern.
    
    Now we can conclude that all the offsprings produced by any operator still follow the pattern in \eqref{eq:def_mt}, which completes the mathematical induction.
    
    Finally, from the above argument, it can be seen that $\Delta W^l$ of the generated model of both crossover and linear interpolation operators must be 0 once the corresponding terms of their parents are 0. Along with the fact that $\Delta W^l=0$ for all the models in the initial generation, we can conclude that $\Delta W^l$ is introduced by mutation if it is non-zero.
\end{proof}

Now, we have already defined the pattern how the target model should follow. Then, we can formulate our optimization problem as:
\begin{equation}
\label{eq:problem_1}
\begin{array}{cl}
\min_{W^l,\theta_i^l,\Delta W^l} &\ell(M) \\ \\
\text{s.t. } & W^l=\sum_{i=1}^n \theta_i^l W_{S,i}^l + \Delta W^l \\ \\
    & \theta_i^l\geq 0,\ \sum_{i=1}^n\theta_i^l=1,
\end{array}
\end{equation}
where $M$ is the model consisting of parameters $\{W^1,\dots,W^L\}$, and $\ell(M)$ is the loss function of model $M$ on the validation data. Any common training criterion for DNN-based acoustic model can be used as the loss function here, such as maximum mutual information (MMI, \cite{bahl1986maximum}).

Finally, due to the different nature between SGD and genetic algorithms, it is much easier for SGD to overfit the validation data when solving \eqref{eq:problem_1}. This is because $\Delta W^l$ can be arbitrary in our current formulation, and it is possible that $\Delta W^l$ becomes large enough to dominate the other terms. To avoid such a problem, we impose an extra restriction on $\Delta W^l$ that its magnitude can not exceed that of all parameters $W^l$ up to a constant factor $\rho\geq 0$, e.g., $\rho=0.01$. As a result, our formal formulation for model merging turns into:
\begin{equation}
\label{eq:problem_2}
\begin{array}{cl}
\min_{W^l,\theta_i^l,\Delta W^l} &\ell(M) \\\\
\text{s.t. } & W^l=\sum_{i=1}^n \theta_i^l W_{S,i}^l + \Delta W^l \\\\
    & \theta_i^l\geq 0,\ \sum_{i=1}^n\theta_i^l=1 \\\\
    & \|\Delta W^l\|\leq \rho\|W^l\|.
\end{array}
\end{equation}

It is possible to solve this problem by using other optimization methods than SGD. But we choose SGD because of two reasons: 1) it is well known that SGD is much faster than traditional optimization methods; 2) our goal is to optimize a DNN, whose objective function is non-convex and has many local minimums, while it is already proven that SGD has stronger ability to escape local minimums \cite{kleinberg2018alternative} and usually converge to better solution in practice.

\subsection{Solving Optimization Problem}

Though we have already formulated our problem into an optimization problem, it is still unclear how to solve it, because this problem seems complicated by having many constraints. In this subsection, we develop a new approach to solve it.

To solve \eqref{eq:problem_2}, one of our basic strategies is that we will not directly update the variable $W^l$ by SGD. Instead, we only update $\theta_i^l$ and $\Delta W^l$, while the value of $W^l$ is inferred from these two groups of variables according to the constraint
\[
W^l=\sum_{i=1}^n \theta_i^l W^l_{S,i} + \Delta W^l
\]
at the beginning of each iteration. And the latest estimation of $W^l$ will work as the bridge for updating $\theta_i^l$ and $\Delta W^l$.

To do updates for $\theta_i^l$ and $\Delta W^l$, we first need to compute the gradients of them. According to the chain rule, the gradient of each $\theta_i^l$ can be derived as:
\[
\frac{\partial \ell}{\partial \theta_i^l}
=\frac{\partial \ell}{\partial W^l}\cdot \frac{\partial W^l}{\partial \theta_i^l}
=\frac{\partial \ell}{\partial W^l}\cdot W^l_{S,i},
\]
where $\cdot$ standards for the dot product of matrices, i.e., $A\cdot B=\sum_i\sum_j A_{ij}B_{ij}$. Therefore, we just need to compute the stochastic gradient of the current model $M=\{W^1,\dots,W^L\}$ by back-propagation as normal, then the computation of the derivative of $\theta_i^l$ becomes trivial from the above equation.
As for $\Delta W^l$, by chain rule again, we can show its gradient is exactly the same as the gradient of $W^l$.

After conducting one step of SGD, we further need to do projection operations to ensure the other two constraints are still satisfied. The constraint
\begin{equation}
    \label{eq:projection_1}
    \theta_i^l\geq 0\quad\text{and}\quad \sum_{i=1}^n\theta_i^l=1,
\end{equation}
is the so-called simplex constraint, which is well-studied in optimization literature, and efficient methods for dealing with it are already known \cite{chen2011projection,condat2016fast}. Hence, we will not dive into the detail of how this projection should be done. Instead, we directly use the projection method proposed in \cite{condat2016fast}, and denote the resulting projection operator as $\Pi(\cdot)$.
While for the constraint
\begin{equation}
    \label{eq:projection_2}
    \left\|\Delta W^l\right\|\leq \rho\left\|W^l\right\|,
\end{equation}
we actually just need to scale $\Delta W^l$ to make it smaller, once we find $\Delta W^l$ violates this constraint. Besides, it is trivial to show that the optimal scaling factor should be:
\[
\gamma_l=\frac{\rho \|W^l\|}{\|\Delta W^l\|}.
\]

Now we have completed all building blocks of our algorithm. The full algorithm is summarized in Algorithm~\ref{alg:sgd}. Here $\eta>0$ is the step size (a.k.a., learning rate) hyperparameter for SGD. In the initialization stage, we just set $\theta_i^l=1/n$ and $\Delta W^l=0$, which implies we choose the model obtained by simple average as the initial model. The motivation for this is the same as including the averaged model in the initial population of GMA: it provides faster convergence and guarantees to produce a better model than the simple average baseline.

\begin{algorithm}[h]
\SetKwInOut{Input}{input}
\SetKwInOut{Output}{output}
\SetAlgoLined
\Input{source models $M_{S,1},M_{S,2},\dots,M_{S,n}$}
    Initialize $\theta_i^l=1/n$ and $\Delta W^l=0$ for all $i$ and $l$ \\
    \While{not converged}{
    \tcp{The following operations are done for all $i$ and $l$}
        Let $W^l=\sum_{i=1}^n \theta_i^l W_{S,i}^l + \Delta W^l$ \\
        Draw a batch of samples from validation data, and compute $\frac{\partial \ell}{\partial W^l}$ by back-propagation \\
        Update by SGD:
        \begin{gather*}
            \theta_i^l = \theta_i^l - \eta \frac{\partial \ell}{\partial W^l}\cdot W_{S,i} \\
            \Delta W^l = \Delta W^l - \eta \frac{\partial \ell}{\partial W^l}
        \end{gather*} \\
        Let $\theta^l=\Pi(\theta^l)$   \tcp*[f]{$\Pi(\cdot)$ is the projection operator for \eqref{eq:projection_1}} \\
        \If{\eqref{eq:projection_2} not hold}{
        Let $\Delta W^l=\Delta W^l\cdot \frac{\rho\|W^l\|}{\|\Delta W^l\|}$
        }
    }
    \Output{$M_T$ with parameters $W^l_T=\sum_{i=1}^n \theta_i^l W_{S,i}^l + \Delta W^l$}
    \caption{SGD-Based Optimizational Merging Algorithm (SOMA)}
    \label{alg:sgd}
\end{algorithm}

GMA is slow because it works in a guess-and-trail fashion: it perturbs models fully random, and adopts performance tests to determine whether the perturbation can bring improvement or not.  However, when dealing with DNN, which has complicated and delicate structures, most of the perturbations are unhelpful. Thus it wastes most of the time in doing useless experiments. As a comparison, SOMA optimizes the model in a guided way: it always moves the model towards the negative gradient direction, which can improve the model at every iteration with high probability.

{
\section{Discussion on Privacy Issues}
Privacy protection is a crucial issue in cooperative training over multiple parties. It is desired by each data curator to guarantee that the other participants cannot know its private training data, neither via direct access nor by indirect inference. Though we do not have a special design for privacy protection in our frame, our methods, including GMA and SOMA, have innate advantages to avoid privacy leakage during cooperative training:
\begin{itemize}
    \item First, it is obvious that our framework need not to share private data, which avoids the direct exposure of privacy.
    \item Our methods aggregate the models in an unpredictable and complicated way. Many privacy attacking methods (e.g., \cite{bagdasaryan2020backdoor,melis2019exploiting}) assume the central server aggregates the clients' information in a simple way such as federated averaging and SGD, which just average the clients' model or gradients for updates. On the contrary, our merging strategies assign dynamic weights to different models and different layers when combining local models, or even mutate the parameters of the merged model. And all these operations are invisible to all the participants except the central server. This nonlinear and blox-box transformation is disastrous for many privacy attackers.
    \item Our methods only require one round of training and merging. To the best of our knowledge, all the existing privacy attacking methods for cooperative training either need to watch the models' changes over different rounds \cite{zhu2019deep,melis2019exploiting,wang2019beyond}, or inject certain backdoor onto the local model to upload and observe other participants' reaction in the following rounds \cite{melis2019exploiting,hitaj2017deep}. While our one-round framework eliminates all these possibilities, making privacy attacks via these means impossible, even for the central server.
\end{itemize}

Of course, there exist other types of privacy attacks, e.g., it is possible to ``reverse engineering'' the final model to recover some information of the training data \cite{ateniese2015hacking}. To reduce this risk, we can include extra protection schemes such as differential privacy (DP, \cite{dwork2014algorithmic}) onto the local training stage. But as a work focusing on cooperative training, we will not discuss the risks out of the training process in depth.
}

\section{Source Model Valuation Based on the Shapley Value}
\label{sec:valuation}

Models' performance on the test set is a widely used criterion for evaluating its quality.  However, for source models used in the merging paradigm, their test set performances cannot represent their contribution to the entire training process. For example, one model may be trained on a data set similar to another data we already have. Even though this model may have low WER, it is not helpful for improving our merged model. Therefore, a fairer way to evaluate a source model is to estimate how much excess benefit it can bring to the generated target model. In our solution, we adopt a cooperative game theory concept, Shapley value \cite{winter2002shapley}, to evaluate each source model for model merging.

We assume that $I = \{M_{S,1},\dots,M_{S,n}\}$ are the entire set of source models that participate in the merge process. The marginal contribution $\phi_i$ of model $M_{S,i}$ is described as the difference between the performance of the target model generated with models $M_{S,i}$ and without $M_{S,i}$:
\begin{equation}
\phi_i = U(S \cup \{M_{S,i}\})- U(S)
\end{equation}
where  $S \subseteq I \backslash \{M_{S,i}\} $, and $U(\cdot)$ represents the utility function. In our setting, $U(S)$ is defined as $1-\text{WER}$ of the model merged from $S$.

Therefore, the Shapley Value of model $M_{S,i}$ is defined as the average marginal contribution $\phi_i$ under all possible subsets $S$ consist of other models in $I$:
\begin{equation}
s_i = \sum_{S \subseteq I \backslash \{M_{S,i}\}}  \frac{1}{n\tbinom{|S|}{N-1}} [ U(S \cup \{M_{S,i}\})- U(S)]
\end{equation}

The complexity of the traditional SV calculation method is exponential and requires repeating a large number of experiments. We follow the method from Jia et al. \cite{jia2019towards} who adopt group testing to estimate the Shapley Value of each source model, significantly reducing the time required for valuation.

In the group testing approach, we first conduct a set of tests. In each test, we randomly select a part of the source model $S \subset I = \{M_{S,1},\dots,M_{S,n}\}$, merge them into a target model using SOMA, and calculate the utility function $U(S)$. Next, we calculate the difference between each of the two source models' utility functions $\Delta U_{ij}$. Finally, due to the group rationality of Shapley value, we can get $\hat{s}$, the estimated Shapley Value, by solving a feasibility problem. The specific steps are shown in Algorithm \ref{alg:sv}.
\begin{algorithm}[h]
\SetKwInOut{Input}{input}
\SetKwInOut{Output}{output}
\SetAlgoLined
\Input{source models $I = \{M_{S,1},\dots,M_{S,n}\}$, the number of tests $T$ }
\textbf{define :} $U(S)=1-\mathrm{WER}(\mathrm{SOMA}(S))$ for any $S\subset I$\\
    Let $Z = 2 \sum_{k=1}^{n-1} \frac{1}{k}$ and $q_k=\frac{1}{Z}\big(\frac{1}{k}+\frac{1}{n-k}\big)$ for $k=1,\dots,n-1$ \\
    Initialize $\beta_{ti} = 0$, for $t=1,\dots,T,i=1,\dots,n$ \\
    \For{$t= 1\ to\ T$}{
        Random sample $l_t$ from $\{1,\dots,n-1\}$ with probability $P(l_t=k)=q_k$ \\
        Sample a subset of source models: $S \subset I $ and $|S|=l_t$ \\
        Let $\beta_{ti} = 1$ for all $M_{S,i} \in S $ \\
        Compute $u_t = U(S)$
    }
    $\Delta U_{ij} = \frac{Z}{T} \sum_{t=1}^T u_t (\beta_{ti}-\beta_{tj}) $ for $i=1,\dots,n,j=1,\dots,n\ $ and $\ j \geq i$ \\
    Find $\hat{s}$ by solving the feasibility problem: \\ $\quad\sum_{i=1}^n \hat{s}_i = U(I)$,\ $|(\hat{s}_i-\hat{s}_j)-\Delta U_{ij}| \leq \frac{\epsilon}{2\sqrt{n}}$,\ $\forall i,j\in \{i,\dots,n\}$  \\
    \Output{the Shapley Value estimation $\hat{s}$ }
    \caption{SV Estimation of source models.}
    \label{alg:sv}
\end{algorithm}

With sufficient number of tests $T$, we can assure $\hat{s}$ is an $(\epsilon,\delta)$-approximation of the true Shapley value $s$, i.e., $P(\|\hat{s} - s\|\leq \epsilon)\geq 1 - \delta$:
\begin{theorem}[Theorem 3 of \cite{jia2019towards}]
Algorithm \ref{alg:sv} returns an $(\epsilon,\delta)$-approximation to the SV if the number of tests $T$ satisfies $T\geq 8\log\frac{n(n-1)}{2\delta}/\left((1 - q_{tot}^2)h\Big(\frac{\epsilon}{Z\sqrt{n}(1-q^2_{tot})}\Big)\right)$, where $q_{tot}=\frac{n-2}{n}q_1+\sum_{k=2}^{n-1}q_k\left[1+\frac{2k(n-k)}{n(n-1)}\right]$ and $h(u)=(1+u)\log(1+u)-u$.
\end{theorem}
By simplifying the bound into the asymptotic form, this theorem implies that $T= \mathcal{O}\left( \frac{n}{\epsilon}\big(\log n\big)^2 \log\frac{1}{\epsilon\delta}\right)$ is enough for an $(\epsilon,\delta)$-approximation. This is much more affordable compared to the naive computation of the SV, which needs $\mathcal{O}(2^n)$ tests.

\section{Experiments}

\subsection{Experimental Setup}

In order to ensure the reproducibility of the experiments, we conduct all experiments on public datasets. Specifically, we collected five speech dataset from the OpenSLR\footnote{\url{http://www.openslr.org/resources.php}} website, which are SLR18, SLR33, SLR38, SLR48 and SLR62. All of them contain Chinese speech recordings in wav format with a sampling rate of 16kHz. Each dataset includes training, validation and test sets \footnote{For dataset which does not have splits in advance, we randomly split it into training, validation and test sets with proportions 60\%: 20\%: 20\%}. The detailed statistics of all the datasets are presented in Table~\ref{table:dataset}.

\begin{table*}[h]
\caption{Statistics of the datasets}
\centering
\large
\begin{tabular}{|c|c|c|c|c|c|c|c|}
\hline
\multirow{2}{*}{Dataset} & \multirow{2}{*}{Name} & \multicolumn{2}{c|}{Training} & \multicolumn{2}{c|}{Validation} & \multicolumn{2}{c|}{Test} \\ \cline{3-8}
                    &      & no. wav     & duration/h     & no. wav        & duration/h       & no. wav     & duration/h    \\ \hline
SLR18     & THCHS-30	               & 7984        & 20.4        & 2657           & 6.7           & 2747        & 7.0        \\ \hline
SLR33     & Aishell               & 120418      & 151.2       & 14331          & 18.1          & 7176        & 10.0       \\ \hline
SLR38      & \begin{tabular}[c]{@{}c@{}}
Free ST Chinese\\Mandarin Corpus
\end{tabular}	              & 61698       & 65.9        & 20395          & 21.8          & 20507       & 22.0       \\ \hline
SLR47     & \begin{tabular}[c]{@{}c@{}}
Primewords Chi-\\nese Corpus Set 1
\end{tabular}               & 30366       & 59.6         & 10092          & 19.8          & 10212       & 20.1       \\ \hline
SLR62     & aidatatang\_200zh	               & 164905      & 139.9       & 24216          & 20.2          & 48144       & 40.2       \\ \hline \hline
Sum           &           & 385371      & 437.1       & 71691          & 86.7          & 88786       & 99.3       \\ \hline
\end{tabular}
\label{table:dataset}
\end{table*}

As a testbed, we develop a full-fledged ASR system through the open-source toolkit Kaldi \cite{povey2011kaldi}. Its built-in ``Chain'' model is used as the acoustic model of the ASR system. The DNN component of the ``Chain'' model is implemented by Time Delay Neural Network (TDNN) \cite{waibel1989phoneme} and the other components of the model such as HMM are pre-trained. The backoff $n$-gram model with $n=3$ is used as the language model, which is trained with the SRILM toolkit \cite{stolcke2002srilm}. The whole system is deployed on a machine with CentOS, Intel Xeon CPU of 72 cores, NVIDIA Tesla K80 GPU and 314GB memory. 

Five TDNNs (i.e., the DNN components of the corresponding ``Chain'' models) with the same initialization are trained on the training sets of the five datasets respectively. The five TDNNs play the roles of the source models. Considering the slow speed of GMA, we sample a subset from the collection of the validation data of the five datasets with a proportion 10\% as the validation set. Both model merging methods will work on this set to optimize the target model by default. All the reported WERs in our experiments are computed on the test set unless otherwise specified.

For the hyperparameters in GMA, we set $K=15$, $p_1=0.5$ and $p_2=p_3=0.1$. The parameter $\rho$ in SOMA is chosen to be $0.1$.

\subsection{Effectiveness Evaluation}

\begin{figure}[h]
\centering
\includegraphics[width=\linewidth]{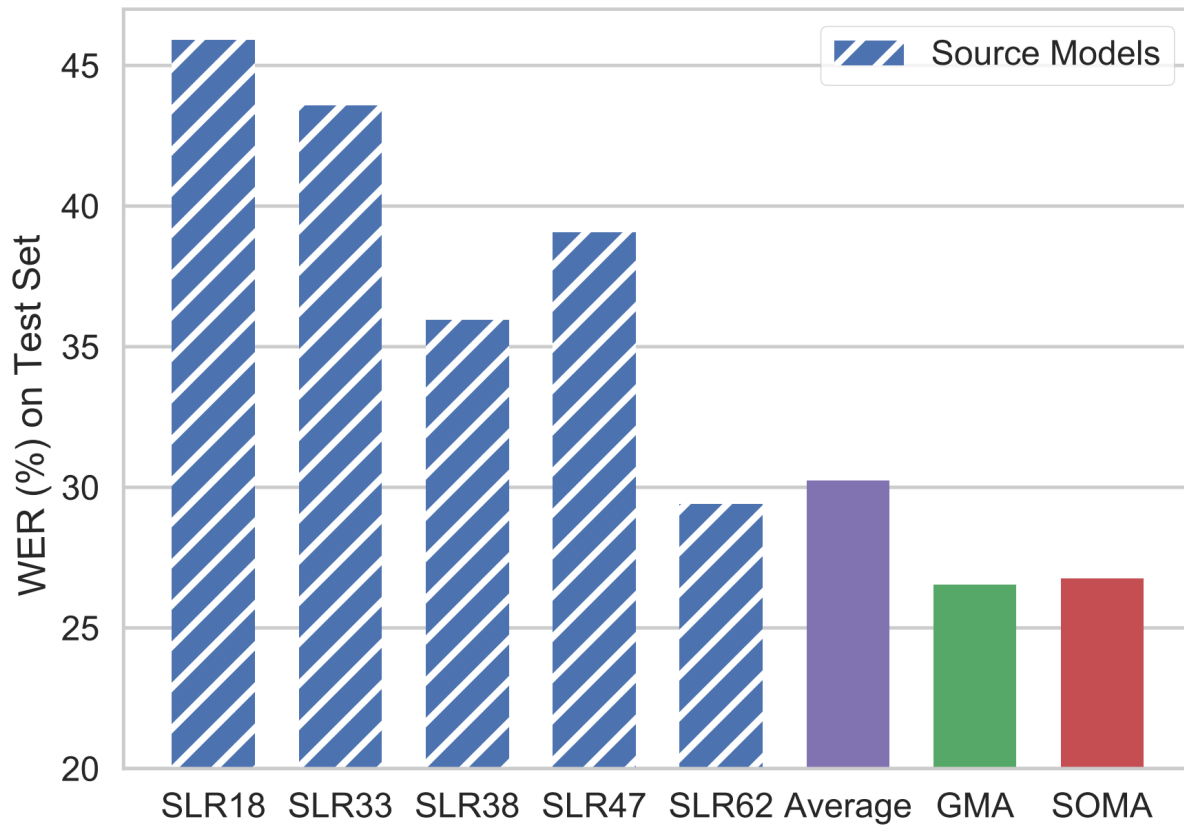}
\caption{WERs of the source models and target models generated by different methods.}
\label{fig:1}
\end{figure}

We first compare the performance of models generated by different methods: direct average, GMA and SOMA. Both GMA and SOMA have been run for enough time until they converge. In detail, SOMA has been run for 10 iterations, where we define one iteration as one passes through the sampled validation data. While for GMA, it has been run for 100 generations with a population size $K=15$, which results in more than 1000 iterations since each generated model needs to be evaluated on the validation set once.

The WERs of all the models, including the source models, are reported in Figure \ref{fig:1}. Due to the different sizes and qualities of the training data, the performances of the source models vary a lot, which brings challenges to model merging. Though direct average achieves a WER lower than most of the source models, it is still slightly worse than the best one. As a comparison, GMA and SOMA obviously outperform all of them. Between them, GMA works better than SOMA, but their difference ($0.2\%$) is quite limited.

\subsection{Efficiency Evaluation}

\begin{figure}[h]
\centering
\includegraphics[width=0.5\textwidth]{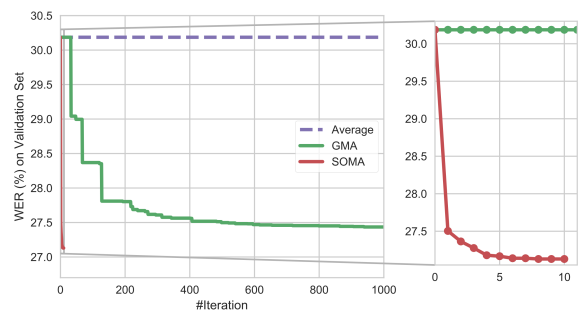}
\caption{Convergence curves of GMA and SOMA.}
\label{fig:2}
\end{figure}

\begin{figure}[h]
\centering
\includegraphics[width=\linewidth]{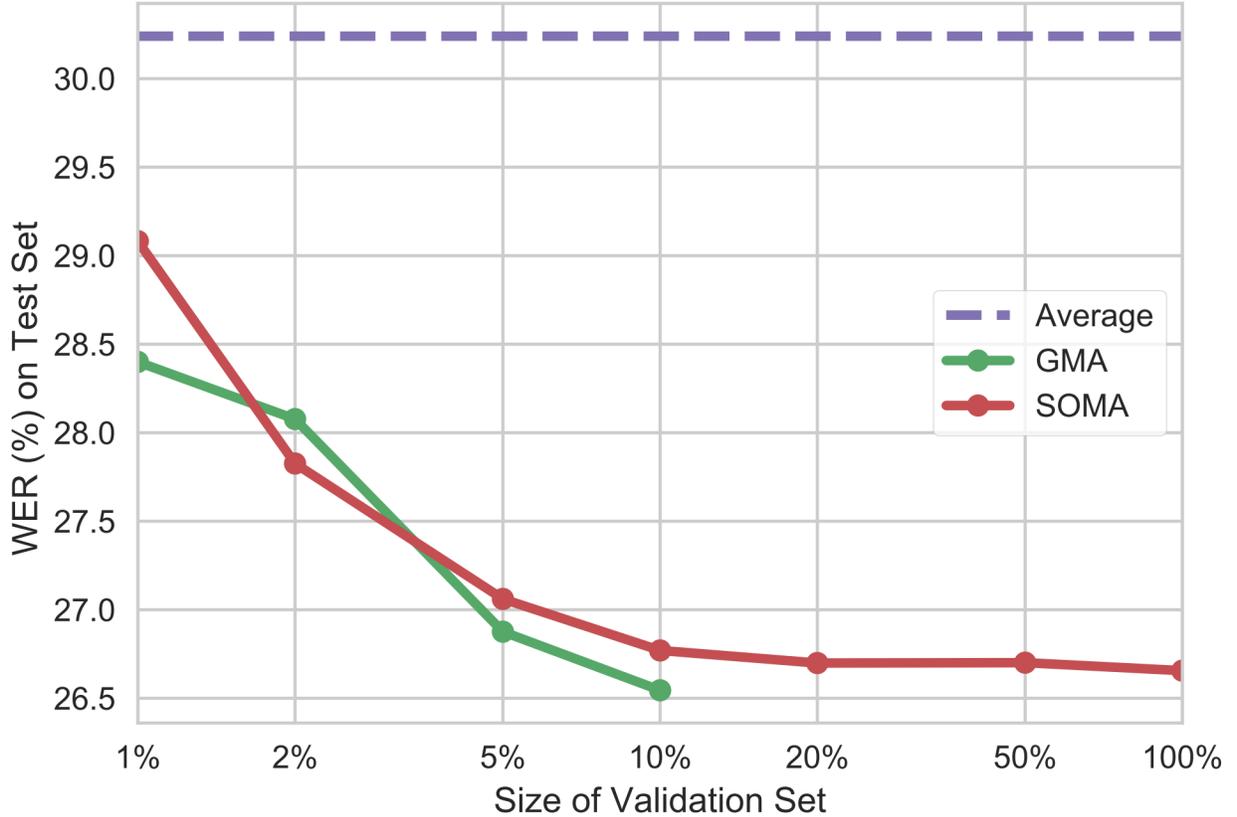}
\caption{WERs of the target models optimized on different sizes of the validation set. Log-scale is adopted for $x$-axis.}
\label{fig:3}
\end{figure}

Though GMA and SOMA have close performance in terms of the generated model quality, they differ substantially in their efficiency. {To demonstrate this, we draw the convergence curves of the first 1000 iterations in Figure \ref{fig:2}, where the WERs of the best-so-far generated models over different iterations are reported. It can be observed that SOMA converges quickly, so that the WER is greatly reduced after just one iteration, and it converges in less than 10 iterations. However, GMA fails to generate a better model than the direct average within the first 30 iterations. It improves the models in a very slow way. Even after 1000 iterations, it still falls behind SOMA. Finally, GMA will become better than SOMA, but it takes a very long optimization time}. Therefore, we can conclude that GMA is impractical on large datasets due to its poor efficiency.

By taking both effectiveness and efficiency into consideration, we recommend using SOMA rather than GMA in practice.

\subsection{Variation of Validation Data Size}

Considering both GMA and SOMA require an extra set of validation data for model merging compared to direct average, in this part we will vary the size of the validation set, and see what will happen to them. We randomly sample subsets of different sizes from the complete validation set with proportions $\{1\%, 2\%, 5\%, 10\%, 20\%, 50\%, 100\%\}$, and run GMA and SOMA on them, and finally evaluate the WERs on the test set. Due to the slow speed of GMA, it is not tested on the validation subsets larger than $10\%$. {The results are reported in Figure \ref{fig:3}.}

Overall, we can observe that larger validation set yields better target models for both GMA and SOMA. However, if the data size is already large enough, say 10\% of the whole validation set, further increasing data volume does not bring too much help for SOMA. Moreover, we can see that both GMA and SOMA can beat the direct average with a very limited number of validation data such as proportion 1\%. Note that 1\% of the validation data only corresponds to approximately 0.2\% of the training data, but can already help reduce test WER for more than one percent compared to the simple average. This fact provides a strong reason for why to use our methods.

 \subsection{Valuation of Source Models}
 \label{sec:exp-value}

\begin{figure}[h]
\centering
\includegraphics[width=\linewidth]{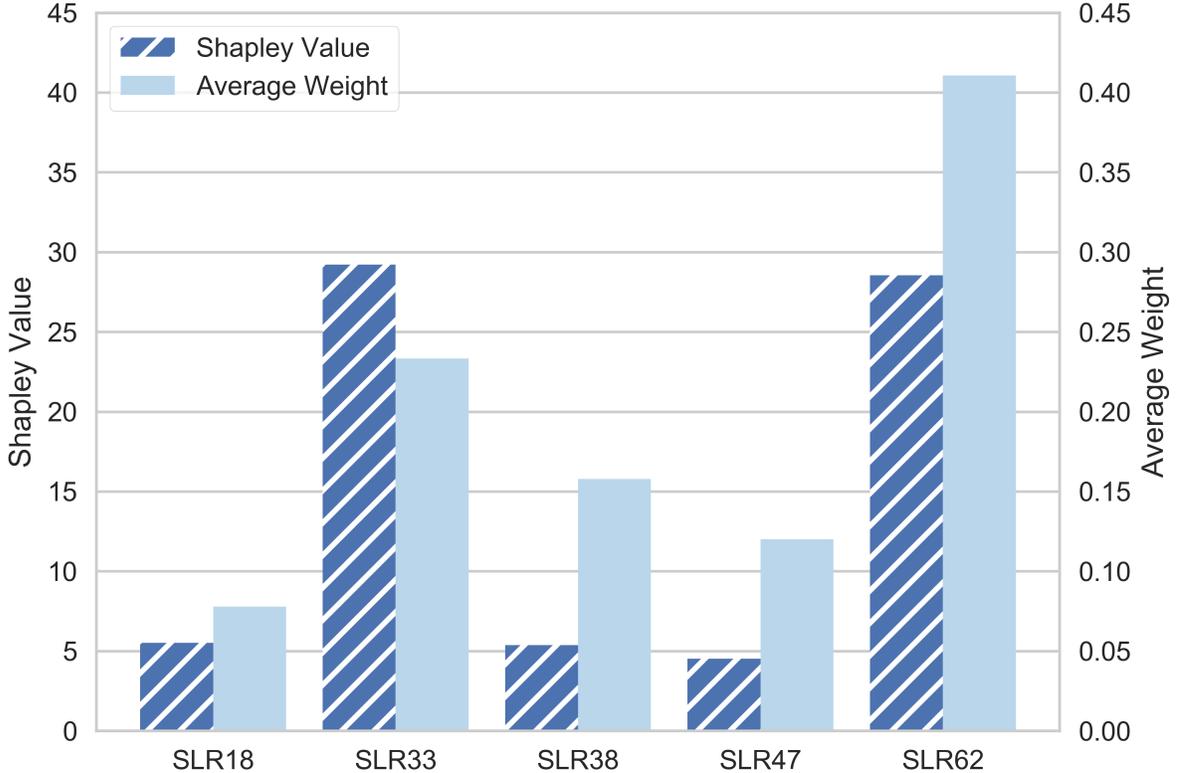}
\caption{Shapley Value and Average Weight of Source Models}
\label{fig:4}
\end{figure}

We conduct experiments with Algorithm \ref{alg:sv} to estimate the Shapley Values (i.e., how much benefit each source model can bring to the target model), which aims to fairly evaluate each source model fine-tuned on different subsets of data. We carried out $T=10$ of group tests to estimate the Shapley Values.

Figure \ref{fig:4} depicts the Shapley Value of each source model. We also report the average weight of each source model $ M_{S,i}$ in the final target model, i.e., $\frac{1}{L} \sum_{l=1}^L \theta_i^l$, which reveals the effect of these source models on the generated target model.

As the results have shown, each source model's Shapley Value is generally positively correlated with its average weight in the target model. However, this is not always the case. The Shapley Value of the source model fine-tuned with SLR33 is a little larger than those fine-tuned with SLR62, but its average weight in the generated target model is smaller.
The benefits that each source model brings to the target model do not entirely depend on its performance on the test set, as is shown in Figure \ref{fig:1}. Although some source models have higher WER, such as the model fine-tuned with SLR33, they can play a more significant role when merging with others.

\section{Conclusion}
In this paper, we propose a novel optimization paradigm for optimizing the acoustic model in ASR in the case where data come from different sources. In our framework, We first train multiple acoustic models independently on distinct parts of data. Then, instead of applying the simplistic averaging scheme for merging acoustic models, we propose two novel algorithms with significantly better performance: GMA and SOMA, where the former is based on a genetic algorithm and the latter one adopts SGD with a novel mathematical formulation. Experiments show that both of them can greatly improve acoustic model quality with very limited amount of validation data. Especially, SOMA demonstrates superior efficiency and can be easily applied to large-scale speech data. Besides, we propose to use the Shapley Value to measure the contribution of different data curators, and empirically study the relationship of Shapley Values to the merging weights.


%

\bibliographystyle{IEEEtran}
\bibliography{IEEEabrv,reference}


\vfill

\end{document}